\documentclass[11pt,a4paper]{article}

\usepackage[margin=1in]{geometry}
\usepackage[T1]{fontenc}
\usepackage[utf8]{inputenc}
\usepackage{lmodern}
\usepackage{microtype}
\usepackage{setspace}
\usepackage{amsmath,amssymb,amsthm}
\usepackage{booktabs}
\usepackage{tabularx}
\usepackage{array}
\usepackage{enumitem}
\usepackage[hidelinks]{hyperref}
\usepackage[nameinlink,noabbrev]{cleveref}

\newtheorem{proposition}{Proposition}[section]
\newtheorem{remark}[proposition]{Remark}

\newcolumntype{Y}{>{\raggedright\arraybackslash}X}

\title{NFT-Based Reward Mechanisms:\\
Sybil Farming, Vesting, and Stochastic Verification}

\author{Marco Alberto Javarone\\
\normalsize Exponential Science Foundation \\
\normalsize m.javarone@exp.science
\and Stefanos Leonardos\thanks{Corresponding author} \\
\normalsize King's College London \\
\normalsize \href{mailto:stefanos.leonardos@kcl.ac.uk}{stefanos.leonardos@kcl.ac.uk}
\and Carmine Ventre \\
\normalsize King's College London \\
\normalsize \href{mailto:carmine.ventre@kcl.ac.uk}{carmine.ventre@kcl.ac.uk}}

\date{}

\begin{document}

\maketitle

\begin{abstract}
We study NFT-based reward mechanisms in which a user can create multiple identities and submit
fraudulent claims that mature a reward subject to vesting. We assume that the issuer stochastically verifies claims during the vesting period and that identities can be linked into clusters so that the detection of one identity submitting a fraudulent claim causes the whole cluster to be forfeited through a penalty. A farmer's payoff is then non-linear in the number of identities: rewards increase linearly, while the probability of avoiding detection decreases geometrically. 
We characterise the optimal farming scale in the continuous relaxation of the problem. This allows us to derive a sufficient condition for deterrence and then consider the
issuer's choice of reward vesting and claim verification. Vesting reduces the probability that a fraudulent claim
is paid but also affects genuine participants, while verification is costly for the issuer. 
We characterise the sufficient deterrence frontier for the issuer in terms of auditing cost, vesting schedule, and penalty value.
When small amounts of audit capacity can be added at negligible marginal cost, vesting alone is not optimal. We also distinguish the role of penalties from that of cluster-level auditing.
Finally, we discuss the implications
for NFT reward programmes on high-throughput ledgers, such as Hedera.
\end{abstract}

\section{Introduction}

NFT-based reward mechanisms have been proposed for loyalty programmes, community participation,
events, digital credentials and creator ecosystems, where a digital record can represent access,
status, provenance or a completed action \cite{bao2022,wang2021,razi2024}. An NFT can encode
attributes, eligibility conditions, access rights or records of contribution, and can therefore be
used as a programmable reward rather than only as a speculative asset
\cite{hartwich2022,taherdoost2023}.

Using NFTs does not by itself solve the problem of rewarding participation. The issuer still controls
eligibility, redemption, expiration and transfer rules. User friction and strategic multi-account
behaviour can also affect the operation of an NFT-based programme
\cite{anjum2022,bcg2023,das2021}. The question we consider is therefore more specific: when do the
properties of an NFT and its underlying ledger help to implement a reward mechanism?

The question is particularly relevant for high-throughput ledgers with low and predictable
transaction costs, such as Hedera \cite{hederadocs}. Frequent reward issuance is feasible on such
platforms, but the same low cost also makes repeated or automated claims inexpensive. Thus, at
sufficient scale, verification and incentives can be more important than the cost of issuing the
reward itself.

We model the farmer's scale as a choice variable. If identities are considered separately, the
farmer's payoff is linear in the number of identities and a profitable attack has no finite scale.
We instead assume that anti-Sybil mechanisms can link identities through signals such as device
information, funding paths or correlated behaviour. Detection of one identity then exposes the
other identities in the same cluster. This cluster-level detection makes the farmer's payoff
non-linear in the number of identities and gives the attack size an economic meaning.

The analysis gives three main results. First, the deterrence condition contains a term that is
absent from an account-level model. An additional identity increases the exposure of the cluster to
detection, and the loss associated with this exposure is the penalty. The term is proportional to
the vesting horizon, the audit hazard and the penalty. Second, we characterise the issuer's choice
of vesting and auditing and show that, when the marginal cost of audit capacity at zero is zero,
vesting alone is not optimal. Third, auditing and penalties have different roles. A penalty makes
detection costly, while cluster-level auditing determines how the addition of identities increases
the probability of detection.

The rest of the paper is organised as follows. \Cref{sec:setting} discusses NFT rewards and the
choice of platform. \Cref{sec:model} introduces the model, and \Cref{sec:deterrence} studies
farming and deterrence. \Cref{sec:design} considers the issuer's problem. \Cref{sec:discussion}
discusses the implications of the results, while \Cref{sec:limitations} considers limitations and
possible extensions.

\section{NFT rewards and the platform choice}
\label{sec:setting}

\subsection{What NFT rewards do and do not provide}

NFTs can represent different types of rewards. For example, a token can record attendance,
completion of a training module, an access right, a contribution or membership
\cite{wang2021,hartwich2022,razi2024}. Such records are useful when the reward is intended to be
persistent or when it has to be recognised by more than one organisation.

It is useful to distinguish the ledger-based record from non-fungibility itself. Conventional
systems can also provide different reward levels, user accounts and custom rules. The advantage of
an NFT is instead that the record can be held and verified across applications or organisations,
provided that the relevant design and governance arrangements support this use
\cite{bao2022,bellucci2025,hartwich2022}.

There are also several limitations. NFTs do not remove the issuer's control over expiration,
transfer, redemption or access, and a smart contract makes the rules auditable without making the
programme decentralised \cite{anjum2022,bellucci2025}. Transferability is a design choice. It can
create resale and arbitrage opportunities and may therefore affect incentives to farm
\cite{fridgen2025,sun2025}. Finally, wallets, transaction signing and other technical requirements
can create additional friction for genuine users \cite{bcg2023,anjum2022}.

In what follows, we treat the NFT as a record of a qualifying action rather than as a speculative
asset. Examples include proof of attendance, training completion, public-sector participation and
rewards for digital contributions \cite{razi2024,nunes2024,martha2023}. The mechanism-design
question is whether genuine activity can be rewarded without making false or duplicated claims
profitable.

\subsection{Platform characteristics}

The choice of ledger depends on transaction costs, throughput, finality, ecosystem maturity,
governance and user experience \cite{wang2021,razi2024}. The relative importance of these features
depends on the application. Frequent low-value rewards require predictable costs and sufficient
operational capacity, while high-value or transferable rewards place more weight on liquidity and
tooling. \Cref{tab:platforms} gives a simple comparison.

\begin{table}[ht]
\centering
\small
\begin{tabularx}{\textwidth}{lYYYY}
\toprule
\textbf{Platform} & \textbf{Fees} & \textbf{Performance} &
\textbf{Ecosystem} & \textbf{Illustrative reward setting} \\
\midrule
Ethereum L1 & High and variable & Lower throughput & Very mature NFT and
DeFi ecosystem & High-value, low-frequency rewards \\
Polygon / EVM L2 & Low & High throughput & Mature EVM tooling &
Large-scale loyalty and engagement \\
Immutable & Very low & High throughput & Gaming-oriented NFT tools &
Game-based rewards \\
Solana & Very low & High throughput and low latency & Growing NFT and gaming
ecosystem & Real-time incentives \\
BNB Chain & Low to medium & High throughput & Strong exchange and retail
integration & Retail loyalty \\
Hedera & Fixed and low & High throughput and fast finality &
Enterprise-oriented, smaller NFT ecosystem & Large-scale enterprise rewards \\
\bottomrule
\end{tabularx}
\caption{Platform characteristics relevant to NFT reward mechanisms.}
\label{tab:platforms}
\end{table}

Hedera, for example, has low and predictable fees, high throughput and fast finality, together with an
enterprise-oriented governance model \cite{hederadocs}. These features are relevant for large
loyalty, event, certification and micro-incentive programmes. Its smaller NFT and DeFi ecosystem
also means that the reward should not rely on secondary-market liquidity; for many of the settings
considered here, a non-transferable record is more natural \cite{razi2024}.

The same low issuance cost can also reduce the cost of fraudulent claims. This is not specific to
Hedera. It motivates the model below, which abstracts from the details of a particular platform and
focuses on the incentives created by a reward paid per identity.

\section{Model}
\label{sec:model}

\subsection{Environment}

An issuer offers a reward NFT of value $R>0$ to each identity that submits a qualifying claim.
The reward vests after $T\ge 0$ periods.

A \emph{genuine participant} performs the intended task at cost $c_G\ge0$ and uses a legitimate
identity at cost $\kappa_0\ge0$. Let $k_G:=c_G+\kappa_0>0$. A \emph{farmer} submits a false,
simulated or low-quality claim at cost $c_F\ge0$ and creates an additional identity at cost
$\kappa_S>0$. Let $k_F:=c_F+\kappa_S>0$. We assume $c_G>c_F$, so that performing the intended
action is more costly than submitting the fraudulent claim.

A claim of type $x\in\{G,F\}$ remains eligible from one period to the next with probability
$q_x\in(0,1)$, independently across periods, where $q_G>q_F$, that is, genuine participants are assumed to be more persistent than farmers. A claim is eligible if it is paid at
vesting; it becomes ineligible if the account is abandoned or the holder no longer satisfies the
programme conditions. Thus, a type-$x$ claim is paid with probability $q_x^T$. 

We consider a stochastic and imperfect verification of claims. In each period of the vesting window, each outstanding
claim is audited independently with probability $a\in[0,1]$; conditional on being audited, a
fraudulent claim is detected with probability $s\in(0,1]$. Eligibility persistence and verification
outcomes are independent. Genuine claims pass verification in our baseline model; false positives are
treated in \Cref{subsec:falsepositives}. Let
\[
\beta:=1-as \in [0,1]
\qquad\text{and, for } as<1,\qquad
L(a):=-\ln\beta=-\ln(1-as)\ \ge 0.
\]
Here $\beta$ is the probability of escaping detection in one period and $L(a)$ is the corresponding
hazard rate. When $as=1$, $\beta=0$ and $L(a)=\infty$; this case is treated separately in
\Cref{rem:perfect}. The use of probabilistic verification is related to Ferraioli and Ventre
\cite{ferraioli2018}.

We use two timing conventions. First, a claim that loses eligibility remains on the ledger and can
still be audited during the vesting window. Second, a penalty applies to a fraudulent claim and
not to the reward payment, so it is incurred even if the claim had already lost eligibility when it
was detected. 

\subsection{Cluster detection}

We assume that the farmer's identities are \emph{linked}. Signals such as device fingerprints,
funding paths and correlated behaviour can connect accounts, so detection of one identity can lead
to the review of the others. We consider the strongest version of this mechanism. If any identity
in the cluster is detected during the vesting window, the whole cluster is forfeited. The farmer
loses all rewards and pays a penalty $F\ge0$, which may represent a seized deposit, exclusion from
future campaigns or loss of a verified credential. Evidence of large-scale malicious activity in
NFT markets provides motivation for considering such correlated behaviour \cite{das2021}.

A farmer choosing $n\ge 0$ identities escapes detection entirely with probability $\beta^{nT}$,
so the expected payoff is
\begin{equation}
\label{eq:farmerpayoff}
U_F(n)\;=\;\underbrace{\beta^{nT}\,n\,q_F^{T}R}_{\text{rewards, if the cluster survives}}
\;-\;\underbrace{\bigl(1-\beta^{nT}\bigr)F}_{\text{expected penalty}}
\;-\;\underbrace{k_F\,n}_{\text{identity and claim costs}} .
\end{equation}
Writing $A:=q_F^{T}R$ for the expected gross value of one identity and $x:=\beta^{T}$ for the
probability that one identity survives the window undetected, \eqref{eq:farmerpayoff} becomes
\begin{equation}
\label{eq:farmercompact}
U_F(n)=x^{n}\left(nA+F\right)-F-k_F n ,
\qquad U_F(0)=0 .
\end{equation}

The reward component is linear in $n$, whereas the probability that the cluster survives falls
geometrically when $x<1$. Hence the attack has a finite scale in this case. If $x=1$, which occurs
when there is no effective detection or no vesting window, the payoff is linear in $n$ and a
profitable attack has no finite scale; see \Cref{subsec:boundary}.

The farmer chooses an integer $n\in\mathbb{N}_0$. A policy \emph{deters} farming if $U_F(n)\le0$ for every admissible $n\ge1$. Thus no positive
number of identities is profitable. We use weak deterrence, so equality corresponds to indifference. 

We first consider the continuous relaxation
$n\in[0,\infty)$, which gives a simple deterrence condition. We then compare it with the integer
problem. The condition from the relaxation is sufficient, but not necessary, for integer deterrence,
so using it as a design constraint is conservative.

\section{Deterrence}
\label{sec:deterrence}

\begin{proposition}[Optimal farming scale and deterrence]
\label{prop:scale}
Suppose $k_F>0$, $T>0$ and $0<as<1$, so that $x=\beta^{T}\in(0,1)$ and $L(a)$ is finite. In the
continuous relaxation:
\begin{enumerate}[noitemsep,label=(\roman*)]
\item $n=0$ is the global optimum if and only if
\begin{equation}
\label{eq:deterrence}
q_F^{T}R \;\le\; k_F \;+\; T\,L(a)\,F .
\end{equation}
\item If \eqref{eq:deterrence} fails, there is a unique $n^{*}>0$ maximising $U_F$, with
$U_F(n^{*})>0$, and $U_F$ is strictly increasing on $[0,n^{*})$ and strictly decreasing on
$(n^{*},\infty)$.
\end{enumerate}
For the integer problem:
\begin{enumerate}[noitemsep,label=(\roman*),start=3]
\item \eqref{eq:deterrence} is sufficient for no integer $n\ge1$ to be profitable, but not
necessary.
\item Suppose \eqref{eq:deterrence} fails, so that $n^{*}$ is well defined by (ii). Then the
farmer's payoff over $\mathbb{N}_0$ is maximised at one of the two integers adjacent to $n^{*}$,
namely $\lfloor n^{*}\rfloor$ or $\lceil n^{*}\rceil$ --- not necessarily the nearer of the two,
since $U_F$ need not be symmetric about its peak. 
Hence no integer
$n\ge1$ is profitable if and only if
$\max\bigl\{U_F(\lfloor n^{*}\rfloor),\,U_F(\lceil n^{*}\rceil)\bigr\}\le 0$. When $n^{*}<1$ we
have $\lfloor n^{*}\rfloor=0$ and $U_F(0)=0$, so the comparison is between $n=0$ and $n=1$ and
the condition reduces to $U_F(1)\le0$. If instead \eqref{eq:deterrence} holds, $n=0$ is already
the integer optimum by (i) and (iii).
\end{enumerate}
\end{proposition}

\begin{proof}
Let $g(n):=x^{n}(nA+F)$, so that $U_F(n)=g(n)-F-k_Fn$ and $\ln x=-TL(a)<0$. Differentiating,
\[
g'(n)=x^{n}\bigl[A-TL(a)(nA+F)\bigr],
\qquad
g''(n)=x^{n}\,TL(a)\bigl[TL(a)(nA+F)-2A\bigr].
\]
Since $nA+F$ is non-decreasing in $n$, $g''$ changes sign at most once, from negative to positive;
hence $g'$ is decreasing and then increasing, with $g'(n)\to 0$ as $n\to\infty$. It follows that
$U_F'(n)=g'(n)-k_F$ is decreasing and then increasing, with limit $-k_F<0$.

Note $U_F'(0)=A-k_F-TL(a)F$, so \eqref{eq:deterrence} is exactly $U_F'(0)\le 0$. If it holds, then
on the decreasing branch $U_F'(n)\le U_F'(0)\le0$, and on the increasing branch $U_F'(n)$ rises towards
$-k_F<0$ and so remains strictly negative. Hence $U_F'(n)<0$ for all $n>0$ and $U_F(n)<U_F(0)=0$.
If instead $U_F'(0)>0$, then $U_F'$ is positive at $0$ and decreasing on its first branch, so it
crosses zero exactly once, at some $n^{*}>0$; on the increasing branch it rises towards $-k_F<0$
and so stays strictly negative, never crossing back. Thus $U_F'(n)>0$ for $n \in [0,n^{*})$ and $U_F'(n)<0$ for
$n\in (n^{*},\infty)$, giving (ii).

For (iii), sufficiency is immediate: if \eqref{eq:deterrence} holds then $U_F(n)<0$ for all real
$n>0$, hence for all integers $n\ge1$. Necessity fails because \eqref{eq:deterrence} may be
violated with $n^{*}<1$, so that the profitable scales are all fractional and unavailable; see
\Cref{rem:integer}. For (iv), when \eqref{eq:deterrence} fails $n^{*}$ exists by (ii) and $U_F$ is
single-peaked on $[0,\infty)$, so its maximum over $\mathbb{N}_0$ is attained at
$\lfloor n^{*}\rfloor$ or $\lceil n^{*}\rceil$; 
since $U_F(0)=0$, no integer
$n\ge1$ is profitable precisely when
$\max\{U_F(\lfloor n^{*}\rfloor),U_F(\lceil n^{*}\rceil)\}\le0$. When \eqref{eq:deterrence} holds
the claim is immediate from (iii).
\end{proof}

\begin{remark}[Perfect detection]
\label{rem:perfect}
If $as=1$ then $\beta=0$, so every submitted claim is detected within the first audit period and
$\beta^{nT}=0$ for $n\ge1$. Hence $U_F(n)=-F-k_Fn<0$ for all $n\ge1$, while $U_F(0)=0$ because a
farmer who creates no identities submits nothing to audit. Farming is deterred for any $F\ge0$.
Condition \eqref{eq:deterrence} cannot deliver this conclusion --- its right-hand side is
$+\infty$ when $F>0$ and indeterminate when $F=0$ --- which is why the case is handled directly.
\end{remark}

\begin{remark}[What the relaxation gives away]
\label{rem:integer}
A difference between the continuous and integer problems can arise only when
\eqref{eq:deterrence} fails and $n^{*}<1$. For example, at $R=147$, $k_F=10$, $q_F=0.92$,
$T=8$, $a=0.49$, $s=0.3$ and $F=22$, we have $U_F'(0)=37.5>0$ and $n^{*}=0.34$, with
$U_F(n^{*})=5.5>0$. Nevertheless, $U_F(1)=-4.7<0$, so no integer scale is profitable. Since
\eqref{eq:deterrence} is sufficient, we use it as the design constraint. This difference is
relevant in the case $F=0$, discussed in \Cref{subsec:boundary}.
\end{remark}

\subsection{Reading the condition}

Condition \eqref{eq:deterrence} compares the expected gross value of an identity, $q_F^TR$, with
its private cost $k_F$ and the additional detection risk $T L(a)F$. The latter term arises because an
additional identity increases the exposure of the whole cluster to detection, while the penalty $F$
is lost if the cluster is detected. This term is important when the private cost of an identity is
small.

Three features follow. First, the instruments enter as a single product, so the
cluster-detection term vanishes if any of $T$, $L(a)$, $F$ is zero; \Cref{rem:compsub} takes this
up. Vesting therefore also provides verification time.

Second, vesting appears on both sides of \eqref{eq:deterrence} and helps on both. It shrinks the
left-hand side geometrically through $q_F^T$, because farmers are less likely than genuine users
to remain eligible; and it grows the right-hand side linearly through $T\,L(a)F$, by lengthening
the audit window. Genuine participants bear only the first effect, and that asymmetry is the
source of vesting's screening power.

Third, the model has a no-detection benchmark as a special case. Setting $a=0$ gives $L=0$ and reduces
\eqref{eq:deterrence} to $q_F^TR\le k_F$; with $T=1$ and $q_F=1$ this gives the familiar condition
$R\le c_F+\kappa_S$. In this case, $U_F(n)=n(A-k_F)$ is linear, so when deterrence fails the
optimal scale is unbounded. An account-level verification model would also give a linear payoff in
$n$, because detection of one identity would not expose the others. The finite optimum obtained here
instead comes from cluster-level detection: any positive detection probability over a positive audit
window, $s>0$, $a>0$ and $T>0$, gives $x<1$.

\begin{remark}[Transferability]
\label{rem:transferability}
If the reward is transferable and can be resold at an exogenous price $p\ge0$, both types value it
at $R+p$. Thus replacing $R$ by $R+p$ in \eqref{eq:deterrence} makes deterrence harder. A higher
horizon or audit rate can then be required. Transferability can also break the link between the
reward and the person who performed the underlying action. For attendance, training, reputation and
eligibility rewards, non-transferability is therefore often the natural choice.
\end{remark}

\section{Designing the mechanism}
\label{sec:design}

\Cref{prop:scale} gives the deterrence condition. We now consider how the issuer can satisfy it
when vesting and verification have different costs.

\subsection{The deterrence frontier}

Let $\mathcal{D}:=\{(a,T)\in[0,1]\times\mathbb{R}_{+}: q_F^TR\le k_F+TL(a)F\}$ be the set of
deterring policies. If $R\le k_F$, farming is already unprofitable at $T=0$. We therefore focus on
$R>k_F$. In this case, vesting alone deters farming for horizons above
\[
\bar{T}:=\frac{\ln\!\left(k_F/R\right)}{\ln q_F}\;>\;0,
\]
the horizon at which the farmer's expected prize has decayed to its cost. Genuine participation
requires $q_G^TR\ge k_G$. Provided $R\ge k_G$ --- otherwise participation fails even at $T=0$ and
no horizon is feasible --- this gives
\[
T\le T_{\max}:=\frac{\ln\!\left(k_G/R\right)}{\ln q_G}\;\ge\;0 .
\]

\begin{proposition}[Structure of the deterrence frontier]
\label{prop:frontier}
Assume $R>k_F$, $F>0$ and $R\ge k_G$. The horizon $T=0$ is excluded and treated in
\Cref{subsec:boundary}, where it is shown that no $T=0$ policy deters when $R>k_F$. For
$T\in(0,\bar{T})$, deterrence requires an audit rate of at least
\begin{equation}
\label{eq:amin}
a_{\min}(T)=\frac{1}{s}\left[1-\exp\!\left(-\frac{q_F^{T}R-k_F}{TF}\right)\right],
\end{equation}
and $a_{\min}(T)=0$ for $T\ge\bar{T}$. Moreover:
\begin{enumerate}[noitemsep,label=(\roman*)]
\item $a_{\min}(T)$ is strictly decreasing for $T \in (0,\bar{T})$.
\item Let
\begin{equation}
\label{eq:tmin}
T_{\min}:=\inf\{T>0:\ a_{\min}(T)\le1\} .
\end{equation}
If $a_{\min}(T)$ exceeds $1$ for small $T$, then $T_{\min}>0$ is the unique solution of
$q_F^{T}R-k_F=T\,F\,L(1)$, we have $a_{\min}(T_{\min})=1$, and the set of feasible horizons is
$[T_{\min},\infty)$. Otherwise $T_{\min}=0$ and the feasible set is $(0,\infty)$, since $T=0$
itself never deters when $R>k_F$.
\item Deterrence is compatible with genuine participation if and only if $T_{\min}\le T_{\max}$,
with the inequality strict when $T_{\min}=0$.
\item Vesting alone deters without deterring genuine users if and only if $\bar{T}\le T_{\max}$,
equivalently $\ln(k_F/R)/\ln q_F \le \ln(k_G/R)/\ln q_G$.
\end{enumerate}
\end{proposition}

\begin{proof}
Consider \eqref{eq:deterrence} with equality, that is, $L(a)=(q_F^TR-k_F)/(TF)$; by definition $L(a)=-\ln(1-as)$, 
and then we get \eqref{eq:amin}, and the requirement is vacuous once $q_F^TR\le k_F$, i.e., $T\ge\bar{T}$.
For (i), write $h(T):=(q_F^TR-k_F)/(TF)$. On $(0,\bar{T})$ the numerator is positive and strictly
decreasing while the denominator is positive and strictly increasing, so $h$ is strictly
decreasing; since $L$ is strictly increasing, $a_{\min}(T)=L^{-1}(h(T))$ inherits strict monotonicity.
For (ii), $a_{\min}$ is continuous and strictly decreasing on $(0,\bar T)$ and zero beyond, so
$\{T>0:a_{\min}(T)\le1\}$ is an interval unbounded above. If its infimum is positive, continuity
gives $a_{\min}(T_{\min})=1$, which written out is $q_F^{T}R-k_F=T\,F\,L(1)$, and the infimum is
attained. If the infimum is zero it is not attained, since $T=0$ does not deter. Parts (iii) and
(iv) restate feasibility of
$\mathcal{D}\cap\{0<T\le T_{\max}\}$ and the special case $a=0$; the restriction to $T>0$ costs
nothing, since $T=0$ never deters under $R>k_F$.
\end{proof}

Part (iv) requires a sufficiently large persistence difference relative to the cost difference.
If genuine users are only slightly more persistent than farmers, vesting alone may destroy
participation before it deters farming.

\begin{remark}[Complements in detection, substitutes in design]
\label{rem:compsub}
The terms complement and substitute refer to two different objects. In the detection technology,
the cluster-deterrence term is $T L(a)F$, so the instruments are complementary in the sense that
the term requires all three. Along the deterrence frontier, however, a higher horizon allows a
lower audit rate, and a larger penalty also allows less auditing. The instruments are then
substitutes. We use the terms in these two senses.
\end{remark}

\subsection{The designer's problem}

We first consider $0<s<1$. By \Cref{prop:frontier}, $T_{\min}>0$. Hence, when the participation
constraint is feasible, an optimum is attained. If $s=1$ and $a=1$, the feasible set has no positive
lower bound and the reduced objective converges to zero as $T\rightarrow0$, while $T=0$ is infeasible.
Thus no optimum is attained in that corner; see \Cref{rem:perfect} and \Cref{subsec:boundary}.

Because $a_{\min}$ is decreasing, the designer faces a one-dimensional menu of deterring policies.
The two instruments are costly in different ways.

Vesting is free to the issuer but taxes genuine participants, who collect only with probability
$q_G^T$; the expected reward value destroyed by a horizon $T$ is $(1-q_G^T)R$ per genuine
participant. Verification costs the issuer directly. Let $c(a)$ be the cost per period of
maintaining the capacity to audit a fraction $a$ of outstanding claims, with $c(0)=0$, $c'(a)\ge0$ and $c''(a)\ge0$. Total verification cost over the vesting window is
$T\,c(a)$.

Note that $c(a)$ is a \emph{capacity} cost and does not vary with the realised number of claims.
An issuer paying per claim verified would instead face a cost growing with claim volume, making
verification spend endogenous to the attack the mechanism is meant to prevent. The capacity
formulation keeps the designer's problem separable from the equilibrium scale of farming; where
audit expenditure is genuinely per-claim, \eqref{eq:designer} below understates the cost of a high audit
rate in large programmes.

Let $\lambda>0$ denote the weight the designer places on the cost borne by genuine participants. For $T>0$, define the reduced objective
\begin{equation}
\label{eq:designer}
\Lambda(T):=\lambda\left(1-q_G^{T}\right)R
+T\,c\bigl(a_{\min}(T)\bigr).
\end{equation}
The designer solves
\[
\min_{0<T\le T_{\max},\ a_{\min}(T)\le1}\Lambda(T).
\]

\begin{proposition}[Optimal mix]
\label{prop:mix}
Let $T^{*}$ be an optimal horizon, when an optimum is attained.
\begin{enumerate}[noitemsep,label=(\roman*)]
\item Suppose the optimum is interior in both instruments: $T^{*}\in(0,\bar{T})$, the constraint
$T\le T_{\max}$ is slack, and $a_{\min}(T^{*})<1$. Then $T^{*}$ satisfies
\begin{equation}
\label{eq:foc}
\underbrace{\lambda R\,q_G^{T}\ln\!\left(1/q_G\right)}_{\text{marginal friction cost of delay}}
\;=\;
\underbrace{-\frac{\mathrm{d}}{\mathrm{d}T}\Bigl[\,T\,c\bigl(a_{\min}(T)\bigr)\Bigr]}_{\text{marginal verification saving}} .
\end{equation}
\item If the audit bound binds at an optimum, $a_{\min}(T^{*})=1$, then necessarily
$T^{*}=T_{\min}$, and \eqref{eq:foc} does not characterise it. This arises when audit capacity is
cheap (see, e.g., $d=5$ in \Cref{tab:compstat}).
\item Suppose $c'(0)=0$. Then $T^{*}<\bar{T}$ strictly and $a_{\min}(T^{*})>0$: relying on vesting
alone is never optimal. This holds in both the interior and the boundary case.
\end{enumerate}
\end{proposition}

\begin{proof}
For (i), substituting $a=a_{\min}(T)$ into \eqref{eq:designer} and differentiating gives
\eqref{eq:foc}; the slackness conditions make the derivative the relevant optimality condition.
For (ii), by \Cref{prop:frontier} the feasible set is $[T_{\min},\infty)$ and $a_{\min}$ is
strictly decreasing, so $a_{\min}(T^{*})=1$ forces $T^{*}=T_{\min}$; there the feasible set is
one-sided and the first-order condition need not hold with equality. For (iii), consider
$T\rightarrow\bar{T}$. By \Cref{prop:frontier}, $a_{\min}(T)\to0$, and $a_{\min}'$ is bounded near
$\bar{T}$: writing $h(T)=(q_F^TR-k_F)/(TF)$ we have $h(\bar T)=0$ and
$h'(\bar{T})=k_F\ln q_F/(\bar{T}F)$, which is finite, so $a_{\min}'(\bar{T})=h'(\bar{T})/s$ is
finite. Hence $c(a_{\min}(T))\to c(0)=0$ and, because $c'(a_{\min}(T))\to c'(0)=0$, also
$T\,c'(a_{\min}(T))\,a_{\min}'(T)\to0$. The right-hand side of \eqref{eq:foc} therefore tends to
zero while the left-hand side tends to $\lambda R\,q_G^{\bar{T}}\ln(1/q_G)>0$. So $\Lambda'(T)>0$
on a left-neighbourhood of $\bar{T}$, giving $T^{*}<\bar{T}$ and, by \eqref{eq:amin},
$a_{\min}(T^{*})>0$.
\end{proof}

The condition $c'(0)=0$ is substantive: it says the first sliver of audit capacity is free at the
margin, which is what makes a vanishingly small audit programme worth running. The quadratic
specification used in the numerical example in Section \ref{sec:example} below satisfies it. If instead $c'(0)>0$, as under a linear cost $c(a)=da$, the
marginal verification saving at $\bar{T}$ tends to the strictly positive limit
$-\bar{T}\,d\,a_{\min}'(\bar{T})$ rather than to zero, and for $d$ large enough this exceeds the
marginal friction cost, so vesting alone \emph{is} optimal: under the calibration of
\Cref{tab:compstat} with $c(a)=da$ substituted, that happens once $d\ge45.8$. Whether some
auditing is always worth doing is therefore a question about the shape of audit capacity cost near
zero, not a general feature of the problem.

The instruments are therefore asymmetric. When $c'(0)=0$, vesting alone is dominated: once the
required audit rate has reached zero, further delay does not improve deterrence but still imposes a
cost on genuine participants. The opposite corner can be optimal when audit capacity is cheap.

\subsection{Boundary cases}
\label{subsec:boundary}

We next consider several boundary cases.

\emph{No vesting window ($T=0$).} With no audit periods the survival factor is $1$; we define this
case directly rather than by evaluating $\beta^{nT}$, which is indeterminate at $as=1$. Then
$U_F(n)=n(R-k_F)$, linear and unbounded whenever $R>k_F$. Audits run \emph{during} the window, so
a programme that pays out immediately has none in which to conduct them and no deterrent beyond
$k_F$. Verification can make the horizon short, but not zero. No $T=0$ policy deters unless $R\le k_F$, in which case the farmer earns
at most zero anyway.

\emph{No effective verification ($a=0$ or $s=0$).} These are economically distinct: under $a=0$
the issuer runs no audits, whereas under $s=0$ it audits and bears the cost but the technology
detects nothing. Both give $L(a)=0$ and collapse \eqref{eq:deterrence} to $q_F^TR\le k_F$, that is
$T\ge\bar{T}$: only vesting deters, and the payoff is linear in $n$, so failure is unbounded. The distinction matters for diagnosis: persistent farming despite an active audit function can
result from insufficient audit intensity $a$ or detection quality $s$, and increasing audit intensity
addresses only the former.

\emph{No penalty ($F=0$).} Here the continuous relaxation and the integer problem come apart most
sharply. Setting $F=0$ gives $U_F'(0)=q_F^TR-k_F$, which does not involve $a$: at $n=0$ there is
nothing to forfeit, so audit intensity has no first-order effect, and the relaxation suggests a
penalty is indispensable for marginal entry deterrence. The integer problem behaves differently.
With $F=0$ the payoff \eqref{eq:farmercompact} becomes $U_F(n)=n\bigl[x^{n}A-k_F\bigr]$, and since
$x<1$ the function in the bracket is decreasing in $n$. Hence no integer $n\ge1$ is profitable if and only if it
is unprofitable at $n=1$:
\begin{equation}
\label{eq:F0}
\beta^{T}q_F^{T}R \;\le\; k_F .
\end{equation}
Condition \eqref{eq:F0} \emph{does} involve $a$, through $\beta=1-as$, and tightens as auditing
intensifies, so sufficiently strong auditing makes even the first identity unprofitable with no
penalty at all. In the calibration below at $T=4$, \eqref{eq:F0} binds at $a=0.564$: below that
rate farming pays, above it the integer problem is deterred outright, even though the continuous
optimum $n^{*}$ remains strictly positive throughout, falling from $1.25$ at $a=0.2$ to $0.21$ at
$a=0.9$ without reaching zero.

Penalties are therefore not dispensable, but their role is narrower than the relaxation implies. A
penalty makes detection consequential at any scale; cluster-level auditing makes each additional
identity expose the operation to whatever consequence exists, and can by itself push the
profitable scale below one whole identity. Where $F=0$ that second channel is all the designer
has, and it works --- but only because identities are indivisible.

\emph{Imperfect detection ($s<1$).} Since $a_{\min}(T)\to 1/s$ as $T$ falls, $a_{\min}(T)\le1$
fails for short horizons whenever $s<1$, giving the feasibility floor $T_{\min}>0$ of
\eqref{eq:tmin}. At $s=1$ the limit is exactly $1$, so the floor collapses to zero: in the
calibration below the smallest feasible horizon falls from $T=1.49$ at $s=0.6$ to essentially zero
at $s=1$. Combined with the ceiling $T_{\max}$, a deterring and participation-compatible mechanism
exists only if $T_{\min}\le T_{\max}$. When detection quality is poor and rewards are large this
interval can be empty: no combination of delay and verification works, and the designer must
instead reduce $R$, raise $\kappa_S$ or improve $s$.

\subsection{Numerical illustration}\label{sec:example}

\Cref{tab:compstat} reports the optimal policy for a calibrated example with quadratic audit
capacity cost $c(a)=\tfrac{d}{2}a^{2}$, taking $R=100$, $k_F=10$, $k_G=30$, $q_G=0.95$,
$q_F=0.85$, $s=0.6$, $F=50$ and $\lambda=1$. This satisfies $c'(0)=0$, so \Cref{prop:mix} applies.
Here $\bar{T}=14.2$ and $T_{\max}=23.5$, so vesting alone is feasible but nonetheless dominated.
At $d=40$ the optimum is $T^{*}=4.15$ with $a^{*}=0.30$. The objective is $26.6$, compared with
$51.7$ under vesting alone, a reduction of roughly one half.

\begin{table}[ht]
\centering
\small
\begin{tabular}{@{}lcc@{\hspace{2em}}lcc@{\hspace{2em}}lcc@{}}
\toprule
\multicolumn{3}{c}{Audit capacity cost $d$} &
\multicolumn{3}{c}{Genuine persistence $q_G$} &
\multicolumn{3}{c}{Penalty $F$} \\
\cmidrule(r){1-3}\cmidrule(r){4-6}\cmidrule(r){7-9}
$d$ & $T^{*}$ & $a^{*}$ & $q_G$ & $T^{*}$ & $a^{*}$ & $F$ & $T^{*}$ & $a^{*}$ \\
\midrule
5   & 1.49 & 1.00 & 0.88 & 2.68 & 0.56 & 10  & 10.20 & 0.14 \\
20  & 2.71 & 0.55 & 0.91 & 3.20 & 0.44 & 25  & 6.68  & 0.22 \\
40  & 4.15 & 0.30 & 0.94 & 3.86 & 0.34 & 50  & 4.15  & 0.30 \\
80  & 5.53 & 0.18 & 0.97 & 4.98 & 0.22 & 100 & 2.22  & 0.39 \\
160 & 6.92 & 0.11 & 0.99 & 6.84 & 0.11 & 200 & 1.05  & 0.49 \\
\bottomrule
\end{tabular}
\caption{Optimal vesting horizon and audit rate as each parameter varies, holding the others at
their baseline values ($R=100$, $k_F=10$, $q_G=0.95$, $q_F=0.85$, $s=0.6$, $F=50$, $d=40$,
$\lambda=1$).}
\label{tab:compstat}
\end{table}

At $d=5$ the optimum $T^{*}=1.49$, $a^{*}=1.00$ is not an interior solution of \eqref{eq:foc}:
audit capacity is cheap enough that the designer pushes the audit rate to its bound and the
horizon down to $T_{\min}$. That corner is informative --- it is the audit-intensive, short-horizon
regime --- but the first-order condition characterises only the remaining rows.

The entries in the table are monotone, but we do not claim these signs for all parameter values.
The corner at $d=5$ also shows that the optimum need not vary smoothly. In this calibration, higher
audit cost gives a longer horizon and a lower audit rate. Higher genuine persistence gives more
vesting, while a larger penalty gives a shorter horizon and a higher audit rate.

The result for $F$ should be interpreted with care. A larger penalty relaxes the deterrence
constraint and can therefore substitute for delay and auditing. In the numerical example it is
associated with a higher audit rate. This is a property of the calibration and is not a general
comparative-static result.

The calibration also illustrates the effect of reducing the audit rate. At $T=4$, the deterrent
rate is $a_{\min}=0.317$. At $a=0.15$, the farmer still has a profitable scale of about one identity;
at $a=0.05$, the optimal scale is about four and a half. At $a=0$ the scale is unbounded. Thus,
reducing verification can have a large effect when the audit rate becomes small.

\subsection{Extension: false positives}
\label{subsec:falsepositives}

The baseline assumes genuine claims always pass verification. Suppose instead that an audited
genuine claim is erroneously rejected with probability $\eta\in[0,1)$, with $\eta<s$. A genuine
participant then collects only if the identity remains eligible and is never falsely rejected, so
individual rationality requires
\begin{equation}
\label{eq:irfp}
q_G^{T}\left(1-a\eta\right)^{T}R \;\ge\; k_G .
\end{equation}

False positives affect the participation constraint as well as the issuer's cost. Under
\eqref{eq:irfp}, genuine participants are affected by auditing at rate $-\ln(1-a\eta)$, while
fraudulent claims face $L(a)=-\ln(1-as)$. The condition $\eta<s$ ensures that auditing is more
likely to reject a fraudulent claim than a genuine one. If $\eta$ approaches $s$, this distinction
disappears and the persistence gap $q_G-q_F$ becomes more important.

We do not re-solve \eqref{eq:designer} with false positives. \Cref{prop:mix} and
\Cref{tab:compstat} apply to the baseline, where auditing has no cost for genuine participants.
Equation \eqref{eq:irfp} gives the modified participation constraint. Characterising the optimal
$(T,a)$ in this extension is left for future work.

\section{Discussion}
\label{sec:discussion}

The results highlight a distinction between the ledger and the incentive mechanism. Low and
predictable transaction costs make frequent issuance possible, but they do not establish that the
underlying action took place. At high volume, the same infrastructure can therefore support both
legitimate participation and fraudulent claims. \Cref{tab:usecases} summarises the main design
responses for the applications discussed in \Cref{sec:setting}.

\begin{table}[ht]
\centering
\small
\begin{tabularx}{\textwidth}{lYY}
\toprule
\textbf{Use case} & \textbf{Principal vulnerability} &
\textbf{Preferred design response} \\
\midrule
Large-scale loyalty &
Duplicate accounts, artificial qualifying transactions, automated claims &
Non-transferable baseline rewards; claim caps; selective audits; delayed vesting for
high-value benefits \\
Public-sector incentives &
Eligibility, legitimacy, accountability, and inappropriate resale &
Verified identity; non-transferability; high-quality verification; transparent audit and
appeal rules \\
Event-based rewards &
False attendance, ticket arbitrage, and bot collection &
Attendance attestations; post-event verification; non-transferable attendance record;
conditional perks \\
Portable reputation &
Status can be purchased rather than earned &
Strict non-transferability; verification of the underlying action; access-based rather
than resale-based value \\
Training and certification &
Credential fraud and privacy concerns &
Non-transferable credentials; issuer verification; privacy-conscious metadata and access
controls \\
Micro-incentives &
Low-quality or automated contributions at scale &
Sampling-based audits; meaningful penalties or deposits; issuance caps; delayed rewards
where retention is valuable \\
\bottomrule
\end{tabularx}
\caption{Mechanism-design implications across reward applications.}
\label{tab:usecases}
\end{table}

The analysis suggests the following design considerations.

\begin{enumerate}[label=\arabic*.]
    \item \textbf{Specify the qualifying action.} The relevant object is the activity that generates
    the reward and the evidence used to verify it. If the activity is costly or difficult to automate,
    the farmer's cost $k_F$ is higher; otherwise stronger verification may be required.

    \item \textbf{Use NFTs where the record is useful.} A persistent and verifiable record can be
    useful when the reward has to be recognised across applications or organisations.

    \item \textbf{Choose transferability carefully.} Non-transferable rewards are natural for
    credentials, reputation, attendance and eligibility. If rewards can be resold, the additional
    resale value increases the return to farming, as discussed in \Cref{rem:transferability}.

    \item \textbf{Choose verification and penalties jointly.} The audit rate, detection quality and
    penalty all affect deterrence. Auditing without a penalty can still deter an integer attack by
    making the first identity unprofitable, but a penalty provides an additional consequence when a
    cluster is detected.

    \item \textbf{Use linkage where possible.} The cluster effect in the model requires identities
    to be linked. Device, funding and behavioural signals can provide such links, but their strength
    is a separate design question.

    \item \textbf{Use vesting where it serves a purpose.} Vesting can distinguish genuine users from
    farmers when their persistence differs and also provides time for audits. Caps are useful as a
    separate way of limiting financial exposure.
\end{enumerate}

Two points should be kept in mind. First, programmes intended for broad participation should keep
wallet, fee and cryptocurrency requirements low, since these enter $k_G$. Second, the instruments
studied here affect incentives rather than technical capability: they can make farming unprofitable,
but they do not make it impossible.

\section{Limitations and future work}
\label{sec:limitations}

The model makes four main assumptions that limit the scope of the results.

\emph{Cluster detection.} Total forfeiture on any single detection is the strongest form of
linkage. It is the right benchmark when the issuer can trace accounts through devices, funding or
behaviour, and the wrong one when identities are genuinely unlinkable --- bought from unrelated
sellers, say, or leaving no signal that would connect them. We conjecture, but do not derive, that
partial linkage would generally weaken the cluster-level deterrent and move the model towards the
account-level benchmark. Depending on how linkage is structured, the deterrent need not retain the
form $T\,L(a)\,F$ at all; establishing what it becomes would show how much deterrent value
survives imperfect attribution, and matters directly for practice, since linkage strength is
itself a design variable.

\emph{Risk neutrality and a collectible penalty.} The farmer maximises expected value and $F$ is
collectible. Where $F$ represents exclusion from future campaigns or reputational loss rather than
a seized deposit, its effective value is uncertain and probably smaller than its nominal value; a
risk-averse farmer, conversely, would be deterred by less. A seized deposit is one way to make $F$ both certain and verifiable.

\emph{An exogenous reward.} We take $R$ as given. Endogenising it would confront the designer with
the tension that a larger reward attracts genuine participation and farming alike: $R$ raises the
left-hand side of \eqref{eq:deterrence} while relaxing genuine individual rationality, so the
optimal $(R,T,a)$ trades participation against the delay and verification needed to protect it.

\emph{A single farmer and an uncapped pool.} With a capped pool, farmers impose congestion on one
another and deterrence becomes strategic among attackers as well as between attacker and designer.
Heterogeneous farmers, differing in $k_F$ or in detection exposure, would replace the single
threshold with type-dependent attack scales, letting the designer trade deterrence of marginal
farmers against tolerance of infra-marginal ones.

The parameters in the model --- $k_F$, $k_G$, $q_F$, $q_G$, $s$, $F$ and the shape of $c(\cdot)$
near zero --- could in principle be estimated from data from a deployed programme. Such estimates
would make it possible to test the conditions in \Cref{prop:scale,prop:frontier,prop:mix}, in
particular whether $c'(0)=0$ is a reasonable description of audit capacity.

\section{Conclusion}

We studied NFT-based reward mechanisms from a platform and mechanism-design perspective. The model
does not imply that NFTs are generally preferable to loyalty points or database-based rewards. Their
main advantage here is the ability to provide a programmable and verifiable record of contribution,
eligibility, access or status. Low-cost, high-throughput ledgers are suitable for high-volume
programmes, but they do not solve the incentive problem. If claiming a reward is cheap, users can
create multiple identities at low cost.

Treating linked identities as a cluster changes the incentive problem. The farmer has a finite
optimal scale, and the deterrence condition contains the additional term $T L(a)F$, which reflects
the exposure created by an extra identity. Vesting and auditing must then be chosen together because
delay affects genuine participants while verification costs the issuer. The analysis also shows that
penalties and audits are not interchangeable. A penalty makes detection costly, while cluster-level
auditing determines how the risk increases with the number of identities. When audit capacity has
zero marginal cost at zero, vesting alone is not optimal.

The main practical implication is that token issuance is only one part of the mechanism. The
programme also needs verification, appropriate transfer rules and incentives that make genuine
participation preferable to strategic extraction.

\section*{Acknowledgements}

This work was supported by UK Research and Innovation (UKRI) through the Engineering and Physical Sciences Research Council (EPSRC) under grant EP/W034042/1 (SAiFE: A Sandbox for AI-powered Trading in Decentralised Financial Markets).

\end{document}